\documentclass{article}
\usepackage{arxiv}
\usepackage[dvipdfmx]{graphicx}
\usepackage[utf8]{inputenc} 
\usepackage[T1]{fontenc}    
\usepackage{hyperref}       
\usepackage{url}            
\usepackage{booktabs}       
\usepackage{amsfonts}       
\usepackage{microtype}      
\usepackage{graphicx}
\usepackage{enumerate}\usepackage{hyperref}
\usepackage{amsfonts}
\usepackage{bussproofs}
\usepackage{mathptmx}
\usepackage{latexsym}
\usepackage{amsmath}
\usepackage{amsthm}
\theoremstyle{remark}
\newtheorem{remark}{Remark}
\theoremstyle{definition}
\newtheorem{definition}{Definition}
\theoremstyle{theorem}
\newtheorem{theorem}{Theorem}
\theoremstyle{theorem}
\newtheorem{lemma}{Lemma}
\theoremstyle{remark}

\theoremstyle{theorem}
\newtheorem{proposition}{Proposition}

\title{Decidability of Quantum Modal Logic}

\newif\ifuniqueAffiliation
\uniqueAffiliationtrue

\ifuniqueAffiliation 
\author{ \href{https://orcid.org/0000-0001-8878-9824}{\hspace{0mm}Kenji ~Tokuo}\\
	Department of Information Engineering, Oita College\\
	National Institute of Technology\\
	Oita 870-0152\\
	Japan\\
	\texttt{tokuo@oita-ct.ac.jp} \\
}
\else
\usepackage{authblk}

\author[1]{%
	\href{https://orcid.org/0000-0000-0000-0000}{\usebox{\orcid}\hspace{1mm}David S.~Hippocampus\thanks{\texttt{hippo@cs.cranberry-lemon.edu}}}%
}
\author[1,2]{%
	\href{https://orcid.org/0000-0000-0000-0000}{\usebox{\orcid}\hspace{1mm}Elias D.~Striatum\thanks{\texttt{stariate@ee.mount-sheikh.edu}}}%
}
\affil[1]{Department of Computer Science, Cranberry-Lemon University, Pittsburgh, PA 15213}
\affil[2]{Department of Electrical Engineering, Mount-Sheikh University, Santa Narimana, Levand}
\fi
\date{January 30, 2025}

\begin{document}
\maketitle
\footnotetext[2]{This is a pre-copyedited, author-produced version of an article accepted for publication in Logic Journal of the IGPL following peer review. The version of record
Kenji Tokuo, Decidability of quantum modal logic, {\it Logic Journal of the IGPL}, 2025
is available online at: https://doi.org/10.1093/jigpal/jzaf010
}

\begin{abstract}
The decidability of a logical system refers to the existence of an algorithm that can determine whether any given formula in that system is a theorem. In this paper, Harrop's lemma is used to prove the decidability of quantum modal logic.
\end{abstract}
\markboth{}
{Decidability of Quantum Modal Logic}
\clearpage

\section{Introduction}
\label{sec:int}

\subsection{Background and Objective}
\label{sub:bac}

\paragraph{QML}
We proposed the quantum modal logic system $\mathbf{QML}$ in the previous paper \cite{Toku2024}. It is an extension of quantum logic known as orthologic, designed to handle general modalities, and is equipped with a Kripke-style semantics and a sequent-based deductive system.
\paragraph{Decidability}
The decidability of a logical system refers to the existence of an algorithm that can determine whether any given formula in that system is a theorem. If a logical system is decidable, it ensures that the search space during the process of inference or proof is finite, thus avoiding unnecessary searches and enabling the efficient derivation of new results.

As is well known, classical propositional logic is decidable. Since there is a finite number of possible truth value assignments for the set of propositional variables that appear in a formula, it is possible to check all cases exhaustively.

Is quantum logic decidable, then? Regarding orthologic, Bruns \cite{Brun1976} provided a decision procedure for ortholattice identities. Since orthologic can be embedded into the classical modal logic system $\mathbf{B}$, the decidability of orthologic also follows from the decidability of $\mathbf{B}$ \cite{Gold1974, Dish1977, Chia2002}. However, the decidability of orthomodular logic remains unknown \cite{Kalm1983}.

As for modal logic, it is known that the minimal normal modal logic $\mathbf{K}$, as well as normal modal logics obtained by adding several axioms from $\mathbf{T}$, $\mathbf{B}$, $\mathbf{4}$, $\mathbf{5}$, and $\mathbf{D}$ to $\mathbf{K}$, are decidable \cite{Chel1980, Harr1958, Gold1987}.

\paragraph{Objective of This Paper}
In this paper, we demonstrate the decidability of $\mathbf{QML}$. Since it is essentially constructed from two modalities---one responsible for the quantum logic part and the other for the modal logic part---and these modalities are intertwined through a forcing relation, its decidability is not trivial.
\subsection{Our Method}
\label{sub:met}
In general, by tracing the sequent calculus from bottom to top, we can conclude that a deductive system is decidable if the complexity of the sequents decreases monotonically. However, since the sequent calculus for $\mathbf{QML}$ includes the cut rule, the complexity of the sequents does not necessarily decrease in an upward direction. Therefore, this approach cannot be applied to $\mathbf{QML}$. Instead, we adopt a method based on Harrop's Lemma \cite{Harr1958}.

\begin{proposition}[Harrop's Lemma]
A logical system that is finitely axiomatizable and has the finite model property is decidable.
\end{proposition}
By exploiting this lemma, we can reduce the task of establishing the decidability of $\mathbf{QML}$ to demonstrating its finite model property.
\section{Quantum Modal Logic}
\label{sec:qml}
This section provides a brief summary of $\mathbf{QML}$ introduced in the previous work \cite{Toku2024}.
\subsection{Language}
\label{sub:lan}
The language of $\mathbf{QML}$ consists of the following symbols.
\begin{itemize}
\item Atomic formulas: $p$, $q$, $\dots$.\footnote{For a technical reason, we assume that the set of all atomic formulas is countable.}
\item Logical connectives: $\wedge$ (conjunction), $\neg$ (negation)
\item Modal operators: $\Box$ (necessity)
\end{itemize}
Formulas are constructed in the usual way. We use lowercase Greek letters, such as $\alpha$, $\beta$, $\dots$, as meta-symbols for formulas. Let $F$ denote the set of all formulas. We use uppercase Greek letters, including those with subscripts, such as $\Gamma$, $\Delta$, $\dots$, and  $\Gamma_1$, $\Gamma_2$, $\dots$, as symbols for subsets of $F$.
\begin{remark}
As usual, we define $\vee$ (disjunction) as $\alpha \vee \beta \equiv \neg(\neg\alpha \wedge  \neg\beta)$ and $\Diamond$ (possibility) as $\Diamond\alpha \equiv \neg\Box\neg\alpha$. 
\end{remark}
The following meta-logical symbols may be employed in proofs for improved readability:
\begin{itemize}
\item $\Rightarrow$: implies
\item $\forall$: for all
\item $\exists$: there exists
\end{itemize}
\subsection{Semantics}
\label{sub:str}
The semantics of $\mathbf{QML}$ is given as follows.
\begin{definition}[Quantum modal structure]
A {\it quantum modal structure} is defined as the following quadruple $\mathcal{S} = \langle W, R_Q, R_M, \rho \rangle$.
\begin{itemize}
\item $W$: a non-empty set.
\item $R_Q$: a reflexive and symmetric relation on $W$.
\item $R_M$: a binary relation on $W$, forced by $R_Q$.
\item $\rho$:  an assignment of an $R_Q$-closed subset of $W$ to each atomic formula.
\end{itemize}
An element of $W$ is called a {\it world} or a {\it state}.
$R_Q$ is called a {\it non-orthogonality relation} on $W$. This name derives from the non-orthogonality relation obtained from the inner product of a Hilbert space.
$R_M$ is called an {\it accessibility relation} with respect to the modal operators $\Box$ and $\Diamond$. This relation is forced by $R_Q$, meaning the following condition must be satisfied:
\[R_M(i,l) \Rightarrow \forall j \in W (R_Q(i,j) \Rightarrow R_M(j,l)).\]
Since $R_Q$ is symmetric, this condition implies that states related by $R_Q$ can see the same worlds with respect to $R_M$.
In the definition of $\rho$, a subset $X$ of $W$ is said to be {\it $R_Q$-closed} if it satisfies the following condition:
\[i \in X \textrm{ iff } \forall j \in W(R_Q(i,j)  \Rightarrow \exists k \in W \textrm{ s.t. }(R_Q(j,k) \textrm{ and } k \in X)).\]
The reason for considering the value of  $\rho$ as an $R_Q$-closed subset rather than an arbitrary subset of $W$ stems from the original idea in quantum logic that the set of all worlds making an experimental proposition true is not merely a set but a closed subspace of a Hilbert space.
\end{definition}
\begin{definition}[Truth]
\label{df:tru}
Let $\mathcal{S} = \langle W, R_Q, R_M, \rho \rangle$ be a quantum modal structure, $i \in W$,  and $\alpha \in F$. We define $\alpha$ to be true at $i$ in $\mathcal{S}$, denoted by $i \models_\mathcal{S} \alpha$, as follows:
\begin{enumerate}
\renewcommand{\labelenumi}{\alph{enumi})}
\item $i \models_\mathcal{S} p$ iff $i \in \rho(p)$ for atomic formulas $p$
\item $i \models_\mathcal{S} \alpha \wedge \beta$ iff $i \models_\mathcal{S} \alpha$ and $i \models_\mathcal{S} \beta$
\item $i \models_\mathcal{S} \neg\alpha$ iff $\forall j \in W (R_Q(i,j) \Rightarrow j \not\models_\mathcal{S} \alpha)$
\item $i \models_\mathcal{S} \Box\alpha$ iff $\forall l \in W (R_M(i,l) \Rightarrow l  \models_\mathcal{S} \alpha)$
\end{enumerate}
\end{definition}
Let $\mathcal{S} = \langle W, R_Q, R_M, \rho \rangle$ be a quantum modal structure.
We write  $\Gamma \models_\mathcal{S} \alpha$ to mean $\forall i \in W (i \models_\mathcal{S} \Gamma \Rightarrow i \models_\mathcal{S} \alpha)$. Here, $i \models_\mathcal{S} \Gamma$ denotes that $\forall \gamma \in \Gamma (i \models_\mathcal{S} \gamma)$. We write $\Gamma \models \alpha$ to mean $\Gamma \models_\mathcal{S} \alpha$ for any quantum modal structure $\mathcal{S}$. Finally, we write $\Gamma \models \Delta$ to mean $\exists \alpha \in \Delta$ s.t. $\Gamma \models \alpha$.
\subsection{Axiomatization}
\label{sub:axi}
An axiomatic system for $\mathbf{QML}$ is presented below. The part concerning non-modal formulas is based on Nishimura \cite{Nish1980}.
Expressions of the form $\Gamma \vdash \Delta$ in axioms and rules are called {\it sequents}, which represent the claim that at least one formula in $\Delta$ can be derived from a finite number of formulas in $\Gamma$.
In general, we write (possibly empty) sets of formulas on the left and right sides of $\vdash$.
We denote $\alpha, \Gamma$ and $\Gamma, \Delta$ as abbreviations for $\{\alpha\} \cup \Gamma$ and $\Gamma \cup \Delta$, respectively. Additionally, we denote $\neg \Gamma$ as $\{\neg \gamma \mid \gamma \in \Gamma\}$ and $\Box \Gamma$ as $\{\Box \gamma \mid \gamma \in \Gamma\}$.
\begin{definition}[Axioms and rules]
\,
\begin{itemize}
\item Axioms
\begin{prooftree}
\AxiomC{$\alpha \vdash \alpha$  \hspace{1.0mm} \scriptsize{(AX)}}
\end{prooftree}
\begin{prooftree}
\AxiomC{$\Gamma \vdash \Box\alpha, \neg\Box\alpha$  \hspace{1.0mm}   \scriptsize{(MEM)}}
\end{prooftree}
\item Rules
\begin{prooftree}
\AxiomC{$\Gamma \vdash \Delta$}
\RightLabel{\scriptsize{(WKN)}}
\UnaryInfC{$\Pi, \Gamma \vdash \Delta, \Sigma$}
\end{prooftree}
\begin{prooftree}
\AxiomC{$\Gamma_1 \vdash \Delta_1, \alpha$}
\AxiomC{$\alpha, \Gamma_2 \vdash \Delta_2$}
\RightLabel{\scriptsize{(CUT)}}
\BinaryInfC{$\Gamma_1, \Gamma_2 \vdash  \Delta_1, \Delta_2$}
\end{prooftree}
\begin{prooftree}
\AxiomC{$\alpha, \Gamma \vdash \Delta$}
\RightLabel{\scriptsize{(${\rm \wedge {\rm l}_1}$)}}
\UnaryInfC{$\alpha \wedge \beta, \Gamma \vdash \Delta$}
\end{prooftree}
\begin{prooftree}
\AxiomC{$\beta, \Gamma \vdash \Delta$}
\RightLabel{\scriptsize{(${\rm \wedge {\rm l}_2}$)}}
\UnaryInfC{$\alpha \wedge \beta, \Gamma \vdash\Delta$}
\end{prooftree}
\begin{prooftree}
\AxiomC{$\Gamma \vdash \Delta, \alpha$}
\AxiomC{$\Gamma \vdash \Delta, \beta$}
\RightLabel{\scriptsize{(${\rm \wedge r}$)}}
\BinaryInfC{$\Gamma \vdash \Delta, \alpha \wedge \beta$}
\end{prooftree}
\begin{prooftree}
\AxiomC{$\Gamma \vdash \Delta, \alpha$}
\RightLabel{\scriptsize{(${\rm \neg l}$)}}
\UnaryInfC{$\neg \alpha, \Gamma \vdash \Delta$}
\end{prooftree}
\begin{prooftree}
\AxiomC{$\alpha \vdash \Delta$}
\RightLabel{\scriptsize{(${\rm \neg r}$)}}
\UnaryInfC{$\neg \Delta \vdash \neg \alpha$}
\end{prooftree}
\begin{prooftree}
\AxiomC{$\alpha, \Gamma \vdash \Delta$}
\RightLabel{\scriptsize{(${\rm \neg \neg l}$)}}
\UnaryInfC{$\neg \neg \alpha, \Gamma \vdash \Delta$}
\end{prooftree}
\begin{prooftree}
\AxiomC{$\Gamma \vdash \Delta, \alpha$}
\RightLabel{\scriptsize{(${\rm \neg \neg r}$)}}
\UnaryInfC{$\Gamma \vdash \Delta, \neg \neg \alpha$}
\end{prooftree}
\begin{prooftree}
\AxiomC{$\Gamma \vdash \alpha$}
\RightLabel{\scriptsize{({\bf K})}}
\UnaryInfC{$\Box\Gamma \vdash \Box\alpha$}
\end{prooftree}
\end{itemize}
\end{definition}
A {\it derivation} is a finite sequence of sequents, where each sequent in the sequence is either an axiom or the lower sequent of a rule, with all upper sequents having already appeared in the sequence.
We say that $\Gamma \vdash \Delta$ is {\it derivable}, or a {\it theorem}, if there exists a derivation where this sequent appears as the last element.
We will simply use $\Gamma \vdash \Delta$ to mean that $\Gamma \vdash \Delta$ is derivable.
\section{Decidability}
\label{sec:dec}
This section demonstrates the decidability of $\mathbf{QML}$.

\begingroup
\def\thetheorem{\ref{prop:fin}}
\begin{proposition}[Finite model property]\label{prop:fin}
If a formula 
$\alpha$ is not satisfied in a quantum modal structure 
$\mathcal{S}$, i.e., 
$\not\models_\mathcal{S} \alpha$, then there exists a quantum modal structure 
$\mathcal{S}^*$  with a finite set of worlds such that 
$\not\models_{\mathcal{S}^*} \alpha$.
\end{proposition}
\addtocounter{proposition}{-1}
\endgroup
We will postpone the proof for now; assuming this proposition, we can conclude that the following main theorem holds.
\begin{theorem}[Decidability]
$\mathbf{QML}$ is decidable.
\end{theorem}
\begin{proof}
Let $\alpha \in F$.
The following two procedures are carried out in parallel.
In the first procedure, we search for a proof diagram for $\alpha$.
Since $F$ is countable, the elements of $F$ can be numbered and arranged in order.
For a positive natural number $k$, we consider sequents constructed from the formulas up to the $k$-th element, where the number of occurrences of each formula does not exceed $k$.
The set of proof diagrams of $\mathbf{QML}$ consisting of such sequents, where the number of occurrences of the sequents is less than or equal to $k$, is denoted by $P_k$.
Since $\mathbf{QML}$ is axiomatizable, $P_k$ is a finite set.
Starting from $P_1$, we check whether a proof diagram for $\vdash \alpha$ is contained.
If $\vdash \alpha$ is provable, the proof diagram for $\alpha$ must be contained in one of the $P_k$, so this procedure will eventually terminate.
In the second procedure, we search for a quantum modal structure that falsifies $\alpha$.
By Proposition \ref{prop:fin}, we specifically need to search for a structure where the set of worlds is finite.
For a positive natural number $k$, only finitely many structures with $|W^*| = k$ exist up to isomorphism.
In light of this, we define $S_k$ as the set of representatives of isomorphism classes.
Starting from $S_1$, we check whether a structure  $\mathcal{S}^*$ exists such that $\exists i \in W^*$ s.t. $i \not\models_{\mathcal{S}^*} \alpha$.
If $\vdash \alpha$ is not provable, then by the completeness of $\mathbf{QML}$ \cite{Toku2024}, a structure falsifying $\alpha$ must be contained in one of the $S_k$, so this procedure will eventually terminate.
\end{proof}
In the following, we will prepare for the proof of Proposition \ref{prop:fin}.
Let $\Sigma$ be a set of formulas satisfying: (i)  If $\alpha \in \Sigma $, then so are subformulas of $\alpha$. (ii) If $p \in \Sigma$, then so is $\neg p$ for atomic $p$. According to Goldblatt \cite{Gold1974}, we call such a set $\Sigma$ {\it admissible}.
Given a quantum modal structure 
$\mathcal{S} = \langle W,R_Q,R_M,\rho \rangle$, we define an equivalence relation 
$\sim$ on $W$ as follows:
$i \sim j$ iff $\forall \alpha \in \Sigma (i \models_{\mathcal{S}} \alpha$ iff $j \models_{\mathcal{S}} \alpha)$.
The equivalence class of 
$i$ with respect to $\sim$ is denoted by $[i]$.
\begin{definition}[Collapse]
For a quantum modal structure $\mathcal{S} = \langle W,R_Q,R_M,\rho \rangle$, the {\it collapse} of $\mathcal{S}$ with respect to $\Sigma$, denoted by $\mathcal{S}^* = \langle W^*,R_Q^*,R_M^*,\rho^* \rangle$, is defined as follows.
\begin{itemize}
\item $W^* = W/\sim$
\item $R_Q^*([i],[j])$ iff $\exists i' \in [i]$ s.t. $(\exists j' \in [j]$ s.t. $R_Q(i',j'))$
\item $R_M^*([i],[l])$ iff $\forall \Box \alpha \in \Sigma (i \models_{\mathcal{S}} \Box\alpha \Rightarrow l \models_{\mathcal{S}} \alpha)$
\item $\rho^*(p) \equiv \{ [i] \in W^* \mid i \in \rho(p) \}$ if $p \in \Sigma$, and $\rho^*(p) \equiv \emptyset$ otherwise.
\end{itemize}
\end{definition}
\begin{lemma}\label{le:qms}
$\mathcal{S}^*$  is a quantum modal structure.
\end{lemma}
\begin{proof}
Indeed,
\begin{itemize}
\item $W^*$: a non-empty set.
\item $R_Q^*$: a reflexive and symmetric relation on $W^*$.
\item $R_M^*$: a binary relation on $W^{*}$, forced by $R_Q^*$.
We can verify that the forcing condition is satisfied as follows.
Suppose $R_M^*([i],[l])$.
Then, we have
$i \models_{\mathcal{S}} \Box\alpha \Rightarrow l \models_{\mathcal{S}} \alpha$
for all $\Box\alpha \in \Sigma$.
Taking the contrapositive, we have  
$l \not\models_{\mathcal{S}} \alpha \Rightarrow i \not\models_{\mathcal{S}} \Box \alpha$.  
Suppose $l \not\models_{\mathcal{S}} \alpha$.
Then, we have  
$i \not\models_{\mathcal{S}} \Box \alpha$.
Furthermore, suppose $R_Q^*([i],[j])$.
Then, we have  $\exists i' \in [i]$ s.t. $(\exists j' \in [j]$ s.t. $R_Q(i',j'))$.
Let $i'$ and $j'$ denote the elements asserted by the existence claim.
By the forcing condition of $R_M$, we have  
$i' \not\models_{\mathcal{S}} \Box \alpha \Rightarrow j' \not\models_{\mathcal{S}} \Box \alpha$.  
Then, we also have  
$i \not\models_{\mathcal{S}} \Box \alpha \Rightarrow j \not\models_{\mathcal{S}} \Box \alpha$.
Therefore,
$j \not\models_{\mathcal{S}} \Box\alpha$.
Thus,
$l \not\models_{\mathcal{S}} \alpha \Rightarrow j \not\models_{\mathcal{S}} \Box\alpha$.
Taking the contrapositive, we have
$j \models_{\mathcal{S}} \Box\alpha \Rightarrow l \models_{\mathcal{S}} \alpha$.
Since $\Box\alpha$ is arbitrary, we have
$\forall \Box\alpha \in \Sigma (j \models_{\mathcal{S}} \Box\alpha \Rightarrow l \models_{\mathcal{S}} \alpha)$,
which means 
$R_M^*([j],[l])$.
\item $\rho^*$:  an assignment of an $R_Q^*$-closed subset of $W^*$ to each atomic formula.
We can verify that the $R_Q^*$-closedness is satisfied as follows.
We need to show that
$[i] \in \rho^*(p)$ iff $\forall [j] \in W^* (R_Q^*([i],[j]) \Rightarrow \exists [k] \in W^*$ s.t. $(R_Q^*([j],[k])$ and $ [k] \in \rho^*(p)))$.
This statement trivially holds when $p \not\in \Sigma$, i.e.,  $\rho^*(p) = \emptyset$.
Thus, we considier the case where $p \in \Sigma$.
Assume that
$\forall [j] \in W^* (R_Q^*([i],[j]) \Rightarrow \exists [k] \in W^*$ s.t. $(R_Q^*([j],[k])$ and $ [k] \in \rho^*(p)))$. 
Let $j \in W$, and suppose $R_Q(i,j)$.
Then, we have $R_Q^*([i],[j])$.
From the assumption, we have $\exists [k] \in W^*$ s.t. $(R_Q^*([j],[k])$ and $ [k] \in \rho^*(p))$.
Let $[k]$ denote the element asserted by the existence claim. Then, we have $R_Q^*([j],[k])$ and $[k] \in \rho^*(p)$.
Here, 
$R_Q^*([j],[k])$ means
$\exists j' \in [j]$ s.t. $(\exists k' \in [k]$ s.t. $R_Q(j',k'))$.
Let $j'$ and $k'$ denote the elements asserted by the existence claim.
Then, 
$[k] \in \rho^*(p)$ implies
$[k'] \in \rho^*(p)$, i.e., 
$k' \models_{\mathcal{S}} p$.
Since $R_Q(j',k')$, we have
 $j' \not\models_{\mathcal{S}} \neg p$.
From this, we have
$j \not\models_{\mathcal{S}} \neg p$,
which means that
$\exists k \in W$ s.t. $(R_Q(j, k)$ and  $k \in \rho(p))$.
Therefore,
$R_Q(i,j) \Rightarrow \exists k \in W$ s.t. $(R_Q(j,k)$ and $ k \in \rho(p))$.
Since $j \in W$ is arbitrary,
we have $\forall j \in W (R_Q(i,j) \Rightarrow \exists k \in W$ s.t. $(R_Q(j,k)$ and $ k \in \rho(p)))$.
Then, we have
$i \in \rho(p)$
by the $R_Q$-closedness of $\rho$.
Therefore,
$[i] \in \rho^*(p)$.
The converse is omitted, as it can be easily derived.
\end{itemize}
\end{proof}
\begin{lemma}\label{le:fin}
Let $\Sigma$ be a finite admissible set of formulas, and suppose 
$|\Sigma| = n$. Then, $|W^*| \leq 2^n$.
	\end{lemma}
\begin{proof}
There are $2^n$ ways to assign truth values to $n$ formulas. 
$W^*$ is the set formed by identifying elements of $W$ that assign the same truth values to the formulas in $\Sigma$. Therefore, $|W^*| \leq 2^n$.
\end{proof}
\begin{lemma}\label{le:tru}
$\forall \alpha \in \Sigma (\forall i \in W ([i] \models_{\mathcal{S}^*} \alpha$ iff $i \models_{\mathcal{S}} \alpha))$.
\end{lemma}
\begin{proof}
We show this by induction on the structure of formulas.
\begin{enumerate}
\renewcommand{\labelenumi}{\alph{enumi})}
\item Case: Atomic formulas. 
Let $p \in \Sigma$ be atomic and $i \in W$.
We have the following equivalences:
$[i] \models_{\mathcal{S}^*} p$ iff
$[i] \in \rho^*(p)$ iff 
$i \in \rho(p)$ iff $i \models_{\mathcal{S}} p$.
\item Case: $\alpha \wedge \beta$.
In this case, $\alpha, \beta \in \Sigma$. Let $i \in W$.
We have the following equivalences:
$[i] \models_{\mathcal{S}^*} \alpha \wedge \beta$ iff
$[i] \models_{\mathcal{S}^*} \alpha$ and $[i] \models_{\mathcal{S}^*} \beta$
iff
$i \models_{\mathcal{S}} \alpha$ and $i \models_{\mathcal{S}} \beta$
iff
$i \models_{\mathcal{S}} \alpha \wedge \beta$.
Here, the second iff follows from the induction hypothesis.
\item Case: $\neg \alpha$.
In this case, $\alpha \in \Sigma$. Let $i \in W$.
First, we will show that $[i] \models_{\mathcal{S}^*} \neg \alpha \Rightarrow i \models_{\mathcal{S}} \neg \alpha$.
Assume $[i] \models_{\mathcal{S}^*} \neg \alpha$.
Furthermore, 
let $j \in W$, and suppose $R_Q(i,j)$.
Then, we have $R_Q^*([i],[j])$.
Hence $[j] \not\models_{\mathcal{S}^*} \alpha$.
By the induction hypothesis, we have
$j \not\models_\mathcal{S} \alpha$.
Therefore,
$R_Q(i,j) \Rightarrow j \not\models_\mathcal{S} \alpha$.
Since $j$ is arbitrary, we have
$\forall j \in W (R_Q(i,j) \Rightarrow j \not\models_\mathcal{S} \alpha)$
which means 
$i \models_{\mathcal{S}} \neg \alpha$.
Next, we will show that 
$i \models_{\mathcal{S}} \neg \alpha \Rightarrow [i] \models_{\mathcal{S}^*} \neg \alpha$.
Assume 
$i \models_{\mathcal{S}} \neg \alpha$.
Furthermore, let
$[j] \in W^*$, and suppose $R_Q^*([i][j])$.
Then, we have
$\exists i' \in [i] $ s.t. $(\exists j' \in [j]$ s.t. $R_Q(i',j'))$.
Let $i'$ and $j'$ denote the elements asserted by the existence claim. Then, we have
$R_Q(i',j')$ and $i' \models_{\mathcal{S}} \neg \alpha$.
Consequently,
$j' \not\models_{\mathcal{S}} \alpha$,
which implies $j \not\models_{\mathcal{S}} \alpha$.
By the induction hypothesis, we have
$[j] \not\models_{\mathcal{S}^*} \alpha$.
Therefore,
$R_Q^*([i],[j]) \Rightarrow [j] \not\models_{\mathcal{S}^*} \alpha$.
Since $[j]$ is arbitrary, we have
$\forall [j] \in W^* (R_Q^*([i],[j]) \Rightarrow [j] \not\models_{\mathcal{S}^*} \alpha)$,
which means $[i] \models_{\mathcal{S}^*} \neg \alpha$.
\item Case: $\Box \alpha$.
In this case, $\alpha \in \Sigma$. Let $i \in W$. 
First, we will show that $[i] \models_{\mathcal{S}^*} \Box \alpha \Rightarrow i \models_{\mathcal{S}} \Box \alpha$.
Assume $[i] \models_{\mathcal{S}^*} \Box \alpha$.
Furthermore, let
$l \in W$, and suppose $R_M(i,l)$
and
$i \models_{\mathcal{S}} \Box\alpha$.
Then, we have
$l \models_{\mathcal{S}} \alpha$.
Hence
$i \models_{\mathcal{S}} \Box\alpha \Rightarrow l \models_{\mathcal{S}} \alpha$,
which means
$R_M^*([i],[l])$.
Consequently,
$[l] \models_{\mathcal{S}^*} \alpha$.
By the induction hypothesis, we have
$l \models_\mathcal{S} \alpha$.
Therefore,
$R_M(i,l) \Rightarrow l \models_\mathcal{S} \alpha$.
Since $l$ is arbitrary,
we have $\forall l \in W (R_M(i,l) \Rightarrow l \models_\mathcal{S} \alpha)$,
which means
$i \models_{\mathcal{S}} \Box \alpha$.
Next, we will show that
$i \models_{\mathcal{S}} \Box \alpha \Rightarrow [i] \models_{\mathcal{S}^*} \Box \alpha$.
Assume $i \models_{\mathcal{S}} \Box \alpha$.
Furthermore,
let
$[l] \in W^*$, and suppose
$R_M^*([i],[l])$.
Then,
we have
$l \models_{\mathcal{S}} \alpha$.
By the induction hypothesis, we have
$[l] \models_{\mathcal{S}^*} \alpha$.
Therefore,
$R_M^*([i],[l]) \Rightarrow [l] \models_{\mathcal{S}^*} \alpha$.
Since $[l]$ is arbitrary, we have
$\forall [l] \in W^* (R_M^*([i],[l]) \Rightarrow [l] \models_{\mathcal{S}^*} \alpha)$,
which means
$[i] \models_{\mathcal{S}^*} \Box \alpha$.
\end{enumerate}
\end{proof}

We restate the first proposition of this section.
\begin{proposition}
If a formula 
$\alpha$ is not satisfied in a quantum modal structure 
$\mathcal{S}$, i.e., 
$\not\models_\mathcal{S} \alpha$, then there exists a quantum modal structure 
$\mathcal{S}^*$  with a finite set of worlds such that 
$\not\models_{\mathcal{S}^*} \alpha$.
\end{proposition}
\begin{proof}
Assume that a formula 
$\alpha$ is not satisfied in a quantum modal structure 
$\mathcal{S} = \langle W,R_Q,R_M,\rho \rangle$ i.e., $\not\models_\mathcal{S} \alpha$.
Let $\Sigma$ be the smallest admissible set such that $\alpha \in \Sigma$.
Note that $\Sigma$ is finite.
In this case, the collapse of  $\mathcal{S}$ with respect to 
$\Sigma$, denoted by
$\mathcal{S}^* = \langle W^*,R_Q^*,R_M^*,\rho^* \rangle$, is a quantum modal structure by Lemma \ref{le:qms}.
Moreover, 
$|W^*|$ is a finite set by Lemma \ref{le:fin}.
Furthermore, $\not\models_\mathcal{S}^* \alpha$
holds by Lemma \ref{le:tru}.
\end{proof}
\section*{Acknowledgements}
We are grateful to the anonymous reviewer for thoroughly examining the manuscript. Their comments have been invaluable in enhancing the clarity of our arguments.
\section*{Funding}
This work was supported by Japan Society for the Promotion of Science (JSPS) KAKENHI [Grant Number JP24K03372].
\end{document}